\documentclass[final,UKenglish,a4paper,thm-restate,autoref,numberwithinsect]{lipics-v2021}

\nolinenumbers                  

\usepackage{amssymb,amsmath,mathrsfs}
\usepackage{mathtools}
\usepackage{thmtools}
\usepackage{xspace}
\usepackage{dsfont}
\usepackage{tikz}
\usetikzlibrary{cd,fit,calc,positioning,arrows}
\usepackage{ifthen}
\usepackage[notref,notcite]{showkeys}
\usepackage[noadjust]{cite}
\usepackage{rotating}
\usepackage{}

\usepackage{ifdraft}

\newcommand{\resetCurThmBraces}{%
\gdef\curThmBraceOpen{(}%
\gdef\curThmBraceClose{)}}
\resetCurThmBraces
\newcommand{\removeThmBraces}{%
\gdef\curThmBraceOpen{}%
\gdef\curThmBraceClose{}}
\resetCurThmBraces

\usepackage{etoolbox}
\patchcmd{\thmhead}{(#3)}{\curThmBraceOpen #3\curThmBraceClose }{}{}

\usepackage{hyperref}
\hypersetup{hidelinks,final,bookmarks}

\addto\extrasUKenglish{
  
}

\usepackage{seqsplit}
\usepackage{xstring}
\newcommand{\defaultshowkeysformat}[1]{%
\StrSubstitute{#1}{ }{\textvisiblespace}[\TEMP]%
\parbox[t]{\marginparwidth}{\raggedright\normalfont\small\ttfamily\(\{\){\color{red!50!black}\expandafter\seqsplit\expandafter{\TEMP}}\(\}\)}%
}

\renewcommand*\showkeyslabelformat[1]{%
\noexpandarg%
\defaultshowkeysformat{#1}%
}

\renewcommand\itemautorefname{Item}
\newcommand{\itemref}[2]{\autoref{#1}.\ref{#2}}

\usepackage[footnote,marginclue,nomargin]{fixme}

\newcommand{\xra}[1]{\xrightarrow{~#1~}}

\newcommand{\smooth}{smooth\xspace}

\DeclareMathOperator{\Sub}{\mathsf{Sub}}

\newcommand{\op}[1]{\operatorname{\mathsf{#1}}}
\newcommand{\id}{\op{id}}

\newcommand{\monoto}{\ensuremath{\rightarrowtail}}

\newcommand{\subto}{\ensuremath{\hookrightarrow}}

\newcommand{\cat}[1]{\mathscr{#1}}
\def\A{\cat A}

\def\C{\cat C}

\newcommand{\E}{\mathcal{E}}

\newcommand{\MM}{\mathcal{S}}
\newcommand{\Set}{\mathsf{Set}}
\newcommand{\Pos}{\mathsf{Pos}}
\newcommand{\Pfn}{\mathsf{Pfn}}
\newcommand{\Gra}{\mathsf{Gra}}
\newcommand{\KVec}{K\text{-}\mathsf{Vec}}

\newcommand{\DCPOb}{\mathsf{DCPO}_\bot}
\newcommand{\CMS}{\mathsf{CMS}}
\newcommand{\MS}{\mathsf{MS}}
\newcommand{\Ord}{\mathsf{Ord}}

\newcommand{\Pow}{\mathscr{P}}

\newcommand{\M}{\mathcal{M}}

\newcommand{\N}{\mathds{N}}

\newcommand{\Z}{\mathds{Z}}

\newcommand{\Coalg}{\mathop{\mathsf{Coalg}}}

\newcommand{\colim}{\mathop{\mathsf{colim}}}

\newcommand{\ter}{\tau}
\newcommand{\ini}{\iota}

\newcommand{\set}[1]{\{#1\}}

\newcommand{\Abar}{B}
\newcommand{\dbar}{d'}
\newcommand{\eps}{\varepsilon}
\newcommand{\opp}{\mathsf{op}}

\newcommand{\mhat}{\widehat{m}}
\newcommand{\ehat}{\widehat{e}}
\newcommand{\chat}{\widehat{c}}  
\renewcommand{\o}{\cdot}
\newcommand{\takeout}[1]{\empty}

\tikzset{shiftarr/.style={
        rounded corners,%
        to path={--([#1]\tikztostart.center)
                     -- ([#1]\tikztotarget.center) \tikztonodes
                     -- (\tikztotarget)},
}}

\theoremstyle{definition}
\newtheorem{defn}[theorem]{Definition} 
\newtheorem{rem}[theorem]{Remark} 

\newtheorem{assumption}[remark]{Assumption}

\title{Initial Algebras Without Iteration}
\authorrunning{J.~Ad\'amek S.~Milius, L.~S.~Moss}
\author{Ji\v{r}\'i Ad\'amek}%
  {Czech Technical University in Prague, Czech Republic and Technische
  Universität Braunschweig}
  {j.adamek@tu-braunschweig.de}%
  {}{Supported by the grant No.~19-0092S of the Czech Grant Agency}
\author{Stefan Milius}%
  {Friedrich-Alexander-Universität Erlangen-Nürnberg, Germany}%
  {stefan.milius@fau.de}%
  {https://orcid.org/0000-0002-2021-1644}
  {Supported by Deutsche Forschungsgemeinschaft (DFG) under project
    \mbox{MI~717/7-1} and as part of the Research and Training Group 2475 ``Cybercrime and Forensic Computing'' (393541319/GRK2475/1-2019)}
\author{Lawrence S.~Moss}%
  {Indiana University, Bloomington IN, USA}%
  {larry.moss@gmail.com}%
  {}{Supported by grant \#586136 from the Simons Foundation.}%

\category{(Co)algebraic pearls}

\Copyright{Ji\v{r}\'i Ad\'amek, Stefan Milius, Lawrence S.~Moss}

\ccsdesc[500]{Theory of computation~Models of computation}
\ccsdesc[500]{Theory of computation~Logic and verification}
\keywords{Initial algebra, Pataraia's theorem, recursive coalgebra,
  initial-algebra chain}

\EventEditors{Fabio Gadducci and Alexandra Silva}
\EventNoEds{2}
\EventLongTitle{9th Conference on Algebra and Coalgebra in Computer Science (CALCO 2021)}
\EventShortTitle{CALCO 2021}
\EventAcronym{CALCO}
\EventYear{2021}
\EventDate{August 31--September 3, 2021}
\EventLocation{Salzburg, Austria}
\EventLogo{}
\SeriesVolume{211}
\ArticleNo{6}

\begin{document}
\FXRegisterAuthor{sm}{asm}{SM}
\FXRegisterAuthor{ja}{aja}{JA}
\FXRegisterAuthor{lm}{alm}{LM}

\maketitle

\begin{abstract}
  The Initial Algebra Theorem by Trnkov\'a et al.~states, under mild assumptions, that an endofunctor has an initial algebra provided it has a pre-fixed point. The proof crucially depends on transfinitely iterating the functor and in fact shows that, equivalently, the (transfinite) initial-algebra chain stops. We give a constructive proof of the Initial Algebra Theorem that avoids transfinite iteration of the functor. For a given pre-fixed point $A$ of the functor, it uses Pataraia's theorem to obtain the least fixed point of a monotone function on the partial order formed by all subobjects of $A$. Thanks to properties of recursive coalgebras, this least fixed point yields an initial algebra. We obtain new results on fixed points and initial algebras in categories enriched over directed-complete partial orders, again without iteration. Using transfinite iteration we equivalently obtain convergence of the initial-algebra chain as an equivalent condition, overall yielding a streamlined version of the original proof.
\end{abstract}

\section{Introduction}
\label{sec:intro}

Owing to the importance of initial algebras in theoretical computer
science, one naturally seeks results which give the existence of
initial algebras in the widest of settings.  We can distinguish two
different, but related ideas which are commonly used in such results.
By Lambek's Lemma, for every endofunctor $F\colon \A\to\A$, every
initial algebra $\alpha\colon FA \to A$ has a structure $\alpha$ which
is an isomorphism.  So one might hope to obtain an initial algebra
from a \emph{fixed point}, viz.~an $F$-algebra with isomorphic
structure. It is sometimes much easier to find a \emph{pre-fixed
  point}, an object~$A$ together with a monomorphism
$m\colon FA \monoto A$. The Initial Algebra Theorem by Trnkov\'a et
al.~\cite{takr} states that, with inevitable but mild assumptions, any
functor $F$ which preserves monomorphisms and has a pre-fixed point
also has an initial algebra.  The proof uses the second prominent idea
in the area: iteration, potentially into the transfinite.  Indeed,
transfinite iteration of $F$ seems to be an essential feature of the
proof.

The purpose of this paper is to prove the Initial Algebra Theorem in
as wide a setting as possible with no use of iteration whatsoever.
Moreover, the side conditions are mild: they apply, e.g.~to the
categories of complete metric spaces and directed-complete partial
orders with a least element (shortly, dcpo with bottom).  To situate
our method in a larger context, recall that some fixed point theorems
are proved with iteration, and some without.  On the iterative side,
we find Kleene's Theorem: continuous functions on $\omega$-cpos with a
least element have least fixed points obtained by iteration in
countably many steps; and Zermelo's Theorem: monotone functions on
chain-complete posets with a least element have least fixed points,
using a transfinite iteration.  On the non-iterative side, we have the
Knaster-Tarski Theorem: monotone functions on complete lattices have
both least and greatest fixed points, obtained by a direct definition
without iteration. A relatively new result is Pataraia's Theorem:
monotone functions on dcpos with bottom have a least fixed point. The
latter two theorems are ordinal-free and indeed constructive.

The initial algebra for a functor $F$ can often be constructed by
iterating the functor, starting with the initial object $0$ and
obtaining a transfinite chain $0\to F0 \to FF0\cdots$
(\autoref{D:chain}). The reason why fixed point theorems are useful
for the proof of the Initial Algebra Theorem is that in every category
$\A$, the collection $\Sub(A)$ of subobjects of a given object $A$ is
a partial order, and the iteration of $F$ can be reflected by a
particular monotone function $f\colon\Sub(A) \to \Sub(A)$ when
$\alpha\colon FA \monoto A$ is a pre-fixed point; it takes a subobject
$u\colon B \monoto A$ to $\alpha\cdot Fu$.  If $\Sub(A)$ is
sufficiently complete, then $f$ has a least fixed point, and we show
that this yields an initial algebra for $F$.

In order to make this step it is important for us that
joins in $\Sub(A)$ are given by colimits in $\A$.
Therefore Pataraia's Theorem is the best choice as a basis for the
move from the least fixed point of $f$ to the initial $F$-algebra. The
reasons are that (a)~it balances the weak assumption of monotonicity
on $f$ with the comparatively weak directed-completeness of the
subobject lattice; and (b)~its use yields an ordinal-free proof (in
contrast to using Zermelo's Theorem, for which~(a) is also the
case). In fact, we present many examples of categories where directed
joins of subobjects are given by colimits, while this is usually not
the case for arbitrary joins, rendering the Knaster-Tarski
Theorem a bad choice for us.

 We start our exposition in \autoref{S:fixed} with
a review of Pataraia's Theorem and also its non-constructive
precursor, Zermelo's Theorem which we use later in
\autoref{S:ini-chain}. The second ingredient for the proof of our main
result are recursive coalgebras, which we tersely review in
\autoref{S:recoalg}. We use the fact that a recursive coalgebra which
is a fixed point already is an initial algebra.  \autoref{S:smooth} discusses the
property that joins of subobjects are given by colimits. We make the
technical notion of \emph{smoothness} parametric in a class $\M$ of
monomorphisms (representing subobjects), and we
prove our results for a  smooth class $\M$. 

Our main result is the new proof of the Initial Algebra Theorem in \autoref{S:ini-thm}.
We apply it in \autoref{S:dcpo} to the category $\DCPOb$ of
dcpos with bottom.  The class of all
embeddings is smooth.  We derive a new result: if an endofunctor
preserves embeddings and has a fixed point, then it has an initial
algebra which coincides with the terminal coalgebra.

Finally, \autoref{S:ini-chain} rounds off our paper by providing the
original Initial Algebra Theorem, which features the
initial-algebra chain obtained by transfinite iteration.  Although our
proof has precisely the same mathematical content as the original one,
it is slightly streamlined in that it appeals to Zermelo's Theorem
rather than unfolding its proof.
\subparagraph*{Related work.}
Independently and at the same time, Pitts and
Steenkamp~\cite{pittsSteenkamp} have obtained a result on the
existence of initial algebras, which makes use of \emph{sized
  functors} and is formalizable in Agda. In effect, they show that a form of 
iteration using sized functors is sufficient to obtain initial algebras.  Our work,
while constructive, is not aimed at formalization, and, as previously
mentioned avoids iteration.%
\smnote{I vote for leaving this half sentence; we did \emph{not} say
  above that this is what differs in our work from Pitts and Steenkamp.}

\section{Preliminaries}
\label{S:prelim}

We assume that readers are familiar with standard notions from the
theory of algebras and coalgebras for an endofunctor~$F$. We denote an
initial algebra for~$F$, provided it exists, by
\[
  \ini\colon F(\mu F) \to \mu F.
\]
Recall that Lambek's Lemma~\cite{lambek} states that its structure
$\ini$ is an isomorphism. This means that $\mu F$ is a \emph{fixed
  point} of $F$, viz.~an object $A \cong FA$.

\subsection{Fixed Point Theorems}
\label{S:fixed}

In this subsection we present preliminaries on fixed point theorems
for ordered structures. The most well-known such results are, of
course, what is nowadays called Kleene's fixed point theorem and the
Knaster-Tarski fixed point theorem. The former is for $\omega$-cpos,
partial orders with joins of $\omega$-chains, with a least element
(\emph{bottom}, for short). Kleene's Theorem states that every
endofunction which is \emph{$\omega$-continuous}, that is preserving
joins of $\omega$-chains, on an $\omega$-cpo has a least fixed
point. The Knaster-Tarski Theorem~\cite{knaster,Tarski55}, makes
stronger assumptions on the poset but relaxes the condition on the
endofunction. In its most general form it states that a monotone
endofunction $f$ on a complete lattice $P$ has a least and greatest
fixed point. Moreover, the fixed points of $f$ form a complete lattice
again.

Here we are interested in fixed point theorems that still work for
arbitrary monotone functions but make do with weaker completeness
assumptions on the poset $P$. One such result pertains to
chain-complete posets.  It should be attributed to
Zermelo, since the mathematical content of the result appears in his
1904 paper~\cite{Zermelo04} proving the Wellordering Theorem.

An \emph{$i$-chain} in a poset $P$ for an ordinal number $i$ is a
sequence $(x_j)_{j < i}$ of elements of $P$ with $x_j \leq x_k$ for
all $j \leq k < i$. 
The poset $P$ is said to be \emph{chain-complete} if every $i$-chain
in it has a join. In particular, $P$ has a least element $\bot$ (take
$i = 0$).

Let $f\colon P\to P$ be a monotone map on the chain-complete poset
$P$. Then we can define an ordinal-indexed sequence $f^i(\bot)$ by the
following transfinite recursion:
\begin{equation}\label{eq-Zerp}
  f^0(\bot) = \bot, \quad f^{j+1}(\bot) = f(f^j(\bot)),
  \quad\text{and}\quad  
  f^j(\bot) = \bigvee_{i < j} f^i(\bot)
  \quad \mbox{for limit ordinals $j$}.
\end{equation}
It is easy to verify that this is a chain in $P$.
\begin{theorem}[Zermelo]\label{T:Zermelo}
  Let $P$ be a chain-complete poset. Every monotone map
  $f\colon P\to P$ has a least fixed point $\mu f$.  Moreover, for
  some ordinal $i$ we have
  \(
    \mu f =  f^i(\bot).
  \)
\end{theorem}
\begin{proof}
  Take $i$ to be any ordinal larger than $|P|$, the cardinality of the
  set $P$. For this $i$, there must be some $j < i$ such that
  $f^j(\bot) = f^{j+1}(\bot)$. Indeed, this follows from Hartogs'
  Lemma~\cite{Hartogs15}, stating that for every set $P$ there exists
  an ordinal $i$ such that there is no injection $i \monoto P$. Thus,
  $f^j(\bot)$ is a fixed point of $f$. Let $f(x) = x$. An easy
  transfinite induction shows that $f^i(\bot) \leq x$ for all
  $i$. Hence, $f^j(\bot)$ is the least fixed point of~$f$.
\end{proof}

There are also variations on \autoref{T:Zermelo}, such as the result
often called the Bourbaki-Witt Theorem~\cite{Bourbaki49,Witt51}; this
states that every inflationary endo-map on a chain-complete poset has
a fixed point above every element. (A map $f\colon P \to P$ is
\emph{inflationary}, if $x \leq f(x)$ for every $x \in P$.)%

\autoref{T:Zermelo} is not constructive.  Our proof
relied on Hartogs' Lemma, which in turn builds on the standard theory
of ordinals. That theory uses classical reasoning. A related point:
some prominent results depending on ordinals are known to be
unavailable in constructive set theory (see~\cite{BauerLumsdane13}).
For many of the end results, there is an alternative, Pataraia's
Theorem~\cite{Pataraia97}, proved without iteration and without
ordinals (see~\autoref{T:Pataraia}). This result is at the heart of
this paper.  It uses dcpos in lieu of chain-complete posets.
 
Pataraia sadly never published his result in written form.  But it has
appeared e.g.~in work by Escard\'o~\cite{Escardo03},
Goubault-Larrecq~\cite{GL13}, Bauer and
Lumsdane~\cite{BauerLumsdane13} based on a preprint by
Dacar~\cite{Dacar09}, and Taylor~\cite{Taylor21}.  We present a proof
based on Martin's presentation~\cite{Martin13}.

First recall that a \emph{directed subset} of a poset $P$ is a
non-empty subset $D \subseteq P$ such that every finite subset of $D$
has an upper bound in $D$. The poset $P$ is called a \emph{dcpo with
  bottom} if it has a least element and every directed subset $D \subseteq P$
has a join. Note that by Markowsky's Theorem~\cite{Markowsky}, a poset
is chain-complete iff it is a dcpo.%
\smnote{I think we should make this remark because it explains nicely,
  that Pataraia and Zermelo have essentially the same content and
  usefulness. Btw, Markowsky's theorem is \emph{not} Iwamura's Lemma
  but the proof of the former uses the latter (see our book) and Jean
  Goubault-Larrec's blog post
  \url{https://projects.lsv.ens-cachan.fr/topology/?page_id=563}}
\begin{rem}\label{R:funorder}\mbox{ }
  \begin{enumerate}
  \item Observe that
 the set of all maps on a poset $P$ form a poset using the
    point-wise order: $f \leq g$ if for every $x \in P$ we have
    $f(x) \leq g(x)$.

  \item Hence, $f\colon P \to P$ is inflationary iff  $\id_P
    \leq f$, where $\id_P$ is the identity function on $P$.
    
  \item\label{R:funorder:3} Function composition is \emph{left-monotone}: we clearly
    have $f \cdot h \leq g \cdot h$ whenever $f \leq
    g$. Right-monotonicity additionally requires that the fixed argument be a
    monotone map: we have $f \cdot g \leq f \cdot h$ for every $g \leq
    h$ whenever $f$ is monotone.

  \item A  monoid $(M, \cdot, 1)$ is \emph{partially ordered} if $M$ carries
    a partial order such that multiplication is monotone: $a \leq
    b$ and $a'\leq b'$ implies $a \cdot a' \leq b \cdot b'$. It is
    \emph{directed complete} if it is a dcpo. An element $z\in M$ is a
    \emph{zero} if $z \cdot m = z = m \cdot z$ for
    every $m \in M$. 
  \end{enumerate}
\end{rem}
\removeThmBraces
\begin{theorem}[{\cite[Thm.~1]{Martin13}}]\label{T:Martin}
  Every directed complete monoid $(M, \cdot, 1)$ whose bottom
  is the unit $1$ has a top element which is a zero.
\end{theorem}
\resetCurThmBraces
\begin{proof}
   The set $M$ itself is directed: for $m, n \in M$ we see
  that $m \cdot n$ is an upper bound since
  \[
    m = m \cdot 1 \leq m \cdot n \geq n \cdot 1 = n,
  \]
  using that $1$ is the bottom and multiplication is monotone. 
  Thus $M$ has a top element $\top = \bigvee M$. We have $\top =
  \top \cdot 1 \leq \top \cdot m$ for every $m \in M$, and clearly $\top \cdot
  m \leq \top$. Thus, $\top \cdot m = \top$ and, similarly $m \cdot
  \top = \top$, whence $\top$ is a zero. 
\end{proof}
%
%
\begin{theorem}[Pataraia's Theorem]\label{T:Pataraia}
  Let $P$ be a dcpo with bottom. Then every monotone map on~$P$ has a
  least fixed point.
\end{theorem} 
%
%
\begin{proof}\mbox{ }
  \begin{enumerate}
  \item Let $M$ the set of all monotone inflationary maps on $P$. This
    is a monoid under function composition, with the unit
    $\id_P$. Furthermore, $M$ is a dcpo with bottom. Indeed, the order
    is the pointwise order from \autoref{R:funorder}, the least
    element is $\id_P$, and directed joins are computed pointwise in
    $P$. Function composition is monotone (in both arguments) by
    \itemref{R:funorder}{R:funorder:3}. By \autoref{T:Martin}, $M$
    therefore has a top element $t\colon P \to P$ which is a zero.
    
  \item\label{T:Pata:2}
    Let $f\colon P \to P$ be inflationary and monotone. Then $f\in M$
    and therefore $f \cdot t = t$. This means that for every
    $x \in P$, $f(t(x)) = t(x)$, whence $t(x)$ is a fixed point of
    $f$.

  \item\label{T:Pata:3} Now let $f\colon P \to P$ be just monotone. Let $\MM$ be the
    collection of all subsets $S$ of $P$ which contain $\bot$, are
    closed under $f$, and under joins of directed subsets. (In more
    detail, we require that if $s\in S$, then $f(s)\in S$; and if
    $X\subseteq S$ is directed, $\bigvee X\in S$.) Clearly, $\MM$ is
    closed under arbitrary intersections. Let $T = \bigcap
    \MM$.
    
    The set of all post-fixed points $x \leq f(x)$ belongs to $\MM$.
    Indeed, $\bot \leq f(\bot)$, and $f(x) \leq f(f(x))$ whenever
    $x \leq f(x)$. Moreover, a join $p = \bigvee D$ of a directed set
    $D$ of post-fixed points of $f$ is post-fixed point: $p$ satisfies
    $d \leq f(d) \leq f(p)$ for every $d \in D$ due to the
    monotonicity of $f$; thus, $p\leq f(p)$. By the minimality of $T$,
    we therefore know that $T$ consists of post-fixed points of
    $f$. Thus, $f$ restricts to a function $f\colon T \to
    T$. That restriction is inflationary (and monotone, of course) and
    therefore has a fixed point $p$ by~\autoref{T:Pata:2}.
    
  \item\label{T:Pata:4} We show that $p$ is a least fixed point of $f\colon P \to
    P$. Suppose that $x$ is any fixed point. The  
    set $L = \set{y \in P : y \leq x}$  belongs to $\MM$.
    Therefore $T \subseteq L$, which
    implies $p\leq x$. \qedhere
  \end{enumerate}
\end{proof}
\begin{corollary}
  The collection of all fixed points of a monotone map on a dcpo with
  bottom forms a sub-dcpo.
\end{corollary}
\noindent
This is analogous to fixed points of a monotone map on a complete
lattice forming a complete lattice again, see Tarski~\cite{Tarski55}.
\begin{proof}
  Let $f$ be monotone on the dcpo with bottom $P$. Put
  $S = \set{x \in P : x = f(x)}$. Suppose that $D \subseteq S$
  be a directed subset, and let $w = \bigvee D$ be its join in
  $P$. Then we have $x = f(x) \leq f(w)$ for every $x \in S$ since $f$
  is monotone. Therefore $w \leq f(w)$ since $w$ is the join of
  $S$. We now see that $f$ restricts to $W = \set{y \in P, w \leq y}$,
  the set of all upper bounds of~$D$ in $P$: for every $y \in W$ we
  have $w \leq f(w) \leq f(y)$, which shows that $f(y) \in
  W$. Moreover, $W$ is clearly a dcpo: it has least element $w$, and
  the join of every directed set of upper bounds of $D$ is an upper
  bound, too. By \autoref{T:Pataraia}, the restriction of $f$ to $W$
  has a least fixed point~$p$, say. In other words, $p$ is the least
  fixed point of $f$ among the upper bounds of~$D$ in~$P$, and
  therefore it is the desired join of $D$ in $S$.
\end{proof}

Here is our statement of a principle which we shall use later as a key step in our main result.
It also appears  in work by Escard\'o~\cite[Thm.~2.2]{Escardo03} and
Taylor~\cite{Taylor21}.
\begin{corollary}[Pataraia Induction Principle]\label{C:PataInd}
Let $P$ be a dcpo
  with bottom. If $f\colon P\to P$ is monotone, then $\mu f$ belongs
  to every subset $S\subseteq P$ which contains $\bot$ and is closed
  under~$f$ and under directed joins.
\end{corollary}
\noindent
This follows from the proof of~\autoref{T:Pataraia}:
items~\ref{T:Pata:3} and~\ref{T:Pata:4} show that $\mu f\in S$.

We apply the above principle to prove the following result that we
will use in \autoref{S:dcpo}. A monotone function $f$ on a dcpo $D$
with bottom is \emph{continuous} if it preserves directed joins, and
\emph{strict} if $f(\bot) = \bot$.
\begin{lemma}\label{L:mu-pres}
  Let $P, Q$ be dcpos with bottom and let $f\colon P \to P$ and
  $g\colon Q \to Q$ be monotone. For every strict continuous map
  $h\colon P \to Q$ such that $g \cdot h = h \cdot f$ we have $h(\mu
  f) = \mu g$.%
  \smnote{No `and' in the first line, please. I'd like to keep this in
    two line as it looks much better this way.}
\end{lemma}
\begin{proof}
  First, $h(\mu f)$ is a fixed point of $g$: we have
  \(
  g (h(\mu f))
  =
  h(f(\mu f))
  =
  h(\mu f).
  \)
  Therefore $\mu g \leq h(\mu f)$. For the reverse, let $S =
  \set{x \in P : h(x) \leq \mu g}$. Since $h$ is strict, we see that
  $\bot \in S$. Moreover, $S$ is closed under $f$, for if $x \in S$ we
  obtain
  \(
  h(f(x)) = g(h(x)) \leq g(\mu g) = \mu g
  \)
  using monotonicity of $g$ in the second step. Finally, $S$ is closed
  under directed joins: if $D \subseteq S$ is a directed set we obtain
  \(
  h(\bigvee D) = \bigvee_{x \in D} h(x) \leq \bigvee_{x\in D} \mu g =
  \mu g,
  \)
  whence $\bigvee D$ lies in $S$. Thus, by \autoref{C:PataInd}, $\mu
  f\in S$, which means that $h(\mu f) \leq \mu g$.
\end{proof}

\subsection{Recursive Coalgebras}
\label{S:recoalg}

A crucial  ingredient for our new proof of the Initial Algebra Theorem
are recursive coalgebras. They are closely connected to well-founded
coalgebras and hence to the categorical formulation of well-founded
induction.%
\smnote{It must be `induction' here; not `recursion'; well-founded
  coalgebras generalize well-founded relations which are what you
  induct on in well-founded induction. The definition of a
  well-founded coalgebra really \emph{is} a (generalized) categorical formulation of
  the proof principle of well-founded induction (over a well-founded relation).} 
 In his work on categorical set theory, Osius~\cite{osius}
first studied the notions of well-founded and recursive coalgebras
(for the power-set 
functor on sets and, more generally, the
power-object functor on an elementary topos). He defined recursive
coalgebras as those coalgebras $\alpha\colon A \to \Pow A$ which have
a unique coalgebra-to-algebra homomorphism into every algebra (see
\autoref{D:recoalg}).

Taylor~\cite{taylor2,taylor3,Taylor21} considered recursive coalgebras
for a general endofunctor under the name `coalgebras obeying the
recursion scheme', and proved the General Recursion Theorem that all
well-founded coalgebras are recursive for more general endofunctors; a
new proof with fewer assumptions appears in recent work~\cite{amm20}.
Recursive coalgebras were also investigated by
Eppendahl~\cite{eppendahl99}, who called them algebra-initial
coalgebras.

Capretta, Uustalu, and Vene~\cite{cuv06} studied recursive coalgebras,
and they showed how to construct new ones from given ones by using
comonads. They also explained nicely how recursive coalgebras allow
for the semantic treatment of recursive divide-and-conquer
programs. Jeannin et al.~\cite{JeanninEA17} proved the
general recursion theorem for polynomial functors on the category of
many-sorted sets; they also provided many interesting examples of
recursive coalgebras arising in programming.

In this section we will just recall the definition and a few basic
results on recursive coalgebras which we will need for our proof of
the initial algebra theorem.
%
\begin{defn}\label{D:recoalg}
  A coalgebra $\gamma\colon C \to FC$ is \emph{recursive} if for every
  algebra $\alpha\colon FA \to A$ there exists a unique
  coalgebra-to-algebra morphism $h\colon C \to A$, i.e.~a
  unique morphism~$h$ such that the square below commutes:
  \begin{equation}\label{diag:coalg-to-alg}
    \begin{tikzcd}
      C 
      \ar{d}[swap]{\gamma} 
      \ar{r}{h}
      &
      A
      \\
      FC
      \ar{r}{Fh}
      &
      FA
      \ar{u}[swap]{\alpha}
    \end{tikzcd}
  \end{equation}
\end{defn}

Recursive coalgebras are regarded as a full subcategory of the
category $\Coalg F$ of all coalgebras for the functor $F$.
\begin{defn}
  A \emph{fixed-point} of an endofunctor is an object $C$ together
  with an isomorphism $C \cong FC$. We consider $C$ both as an algebra
  and a coalgebra for $F$. 
\end{defn}

\removeThmBraces
\begin{rem}[{\cite[Prop.~7]{cuv06}}]\label{R:iso-ini}
  Every recursive fixed point is an initial algebra:
  for a coalgebra $(C,\gamma)$ with $\gamma$ invertible, the
  coalgebra-to-algebra morphisms from $(C,\gamma)$ to an algebra
  $(A,\alpha)$ are the same as the algebra homomorphisms from
  $(C,\gamma^{-1})$ to $(A,\alpha)$.
\end{rem}
\resetCurThmBraces
%
%
\removeThmBraces
\begin{proposition}[{\cite[Prop.~6]{cuv06}}]\label{P:recFappl}
  If $(C, \gamma)$ is a recursive coalgebra, then so is $(FC, F\gamma)$. 
\end{proposition}
\resetCurThmBraces
\begin{proof}
  Let $(A,\alpha)$ be an algebra and denote by $h\colon C \to A$ the
  unique coalgebra-to-algebra morphism. We will show that
  \[
    g = \big(FC \xra{Fh} FA \xra{\alpha} A\big) 
  \]
  is the unique coalgebra-to-algebra morphism from $(FC,F\gamma)$ to
  $(A,\alpha)$. First, diagram~\eqref{diag:coalg-to-alg} for $g$
  commutes as can be seen on the left below:
  \[
    \begin{tikzcd}
      FC
      \ar{r}{Fh}
      \ar{d}[swap]{F\gamma}
      \ar[shiftarr = {yshift=20}]{rr}{g}
      &
      FA
      \ar{r}{\alpha}
      &
      A
      \\
      FFC
      \ar{r}{FFh}
      \ar[shiftarr = {yshift=-18}]{rr}{Fg}
      &
      FFA
      \ar{u}[swap]{F\alpha}
      \ar{r}{F\alpha}
      &
      FA
      \ar{u}[swap]{\alpha}
    \end{tikzcd}
    \qquad\qquad
    \begin{tikzcd}
      C
      \ar{r}{\gamma}
      \ar[shiftarr ={yshift=20}]{rr}{k \cdot \gamma}
      \ar{d}[swap]{\gamma}
      &
      FC
      \ar{r}{k}
      \ar{d}{F\gamma}
      &
      A
      \\
      FC
      \ar{r}{F\gamma}
      \ar[shiftarr = {yshift=-18}]{rr}{F(k \cdot \gamma)}
      &
      FFC
      \ar{r}{Fk}
      &
      FA
      \ar{u}[swap]{\alpha}
    \end{tikzcd}
  \]
  To see that $g$ is unique, suppose that $k\colon FC \to A$ is a
  coalgebra-to-algebra morphism from $(FC,F\gamma)$ to
  $(A,\alpha)$. Then $k \cdot \gamma\colon C \to A$ is one from
  $(C,\gamma)$ to $(A,\alpha)$. This is shown by the diagram on the
  right above. Thus, we have $h = k \cdot \gamma$, and we conclude
  that
  \[
    g = \alpha \cdot Fh = \alpha \cdot Fk \cdot F\gamma = k,
  \]
  where the last equation holds since $k$ is a coalgebra-to-algebra morphism.
\end{proof}
\begin{corollary}\label{C:Lambek-cor}
  If a terminal recursive $F$-coalgebra exists, it is a fixed point of~$F$.
\end{corollary}
\noindent
Indeed, the proof is the same as that for Lambek's Lemma,
using \autoref{P:recFappl} to see that for a terminal recursive
coalgebra $(T,\ter)$, the coalgebra $(FT,F\ter)$ is recursive, too: the unique
coalgebra homomorphism $h\colon(FT, F\ter) \to (T\ter)$ satisfies $h
\cdot \ter = \id_T$ since $\ter\colon (T,\ter) \to (FT,F\ter)$ is a
coalgebra homomorphism, and finally, $\ter \cdot h = Fh \cdot F\ter =
F\id_T = \id_{FT}$. 
\removeThmBraces
\begin{theorem}[{\cite[Prop.~7]{cuv06}}]
 \label{T:rec-muF}
  The terminal recursive coalgebra is precisely the same as the initial
  algebra.
\end{theorem}
\resetCurThmBraces
\noindent
In more detail, let $F\colon\A \to \A$ be an endofunctor. Then we have:
\begin{enumerate}
\item If $(T, \ter)$ is a terminal recursive coalgebra, then
  $(T, \ter^{-1})$ is a initial algebra.
\item If $(\mu F, \ini)$ is an initial algebra, then $(\mu F,
  \ini^{-1})$ is a terminal recursive coalgebra.
\end{enumerate}
\begin{proof}\mbox{ }
  \begin{enumerate}
  \item By \autoref{C:Lambek-cor}, we know that $\tau$ is an
    isomorphism. By \autoref{R:iso-ini}, $(T,\ter^{-1})$ is an initial
    algebra.
  
  \item The coalgebra $(\mu F, \ini^{-1})$ is clearly
    recursive. It remains to verify its terminality. So let
    $(C, \gamma)$ be a recursive coalgebra.  There is a unique
    coalgebra-to-algebra morphism from $(C,\gamma)$ to the algebra
    $(\mu F, \ini)$ to $(A,\alpha)$, and this means that there is a
    unique coalgebra homomorphism
    $h\colon (C,\gamma)\to (\mu F, \ini^{-1})$.\qedhere
  \end{enumerate}
\end{proof}
\begin{proposition}\label{P:colimrec}
  Every colimit of recursive coalgebras is recursive.
\end{proposition}
\begin{proof}
  We use the fact that the colimits in $\Coalg F$, the category of
  coalgebras for $F$, are formed on the level of the underlying
  category.  Suppose that we are given a diagram of recursive
  coalgebras $(C_i, \gamma_i)$, $i \in I$, with a colimit cocone
  $c_i\colon (C_i, \gamma_i) \to (C,\gamma)$ in $\Coalg F$. We prove
  that $(C,\gamma)$ is recursive, too. Indeed, given an algebra
  $(A,\alpha)$ one takes for every $i$ the unique coalgebra-to-algebra
  morphisms $h_i\colon (C_i,\gamma_i) \to (A,\alpha)$. Using unicity
  one sees that all $h_i$ form a cocone of the diagram formed by all
  $C_i$ in the underlying category. Therefore, there is a
  unique morphism $h\colon C \to A$ such that $h \cdot c_i = h_i$
  holds for all $i \in I$. We now verify that $h$ is the desired
  unique coalgebra-to-algebra morphism using the following diagram:
  \[
    \begin{tikzcd}
      C_i
      \ar{r}{c_i}
      \ar{d}[swap]{\gamma_i}
      \ar[shiftarr = {yshift=20}]{rr}{h_i}
      &
      C
      \ar{d}{\gamma}
      \ar[dashed]{r}{h}
      &
      A
      \\
      FC_i
      \ar{r}{Fc_i}
      \ar[shiftarr = {yshift=-18}]{rr}{Fh_i}
      &
      FC
      \ar{r}{Fh}
      &
      FA
      \ar{u}[swap]{\alpha}
    \end{tikzcd}
  \]
  We know that the upper and lower parts, the left-hand square and the
  outside commute. Therefore so does the right-hand square when
  precomposed by every $c_i$. Since the colimit injections $c_i$ form
  a jointly epic family, we thus see that the right-hand square
  commutes if and only if $h \cdot c_i = h_i$ holds for all $i \in I$. 
\end{proof}

\section{Smooth Monomorphisms}
\label{S:smooth}

As we have just seen in \autoref{P:colimrec}, the collection of recursive
coalgebras is closed under colimits.  In order to apply an
order-theoretic fixed point theorem to this collection, or to
subcollections of it, we need a connection between colimits and
subobjects.  We make this connection by using the definition of
\emph{smooth class} of monomorphisms in a category.

For an object~$A$ of a category $\A$, a \emph{subobject} is
represented by a monomorphism $s\colon S \monoto A$. If $s$ and
$t \colon T \monoto A$ are monomorphisms, we write $s \leq t$ if $s$
factorizes through $t$. If also $t \leq s$ holds, then $t$ and $s$
represent the same subobject; in particular $S$ and $T$ are then
isomorphic.
Generalizing a bit, let $\M$ be a class of monomorphisms. An
\emph{$\M$-subobject} of $A$ is a subobject represented by a morphism
$s\colon S\to A$ in $\M$. If the object $A$ has only a set of
subobjects, then we write
\[
  \Sub_{\M}(A)
\]
for the poset of $\M$-subobjects of $A$.

If every object $A$ only has a set of $\M$-subobjects, then $\A$ is
called \emph{$\M$-well-powered}.

\begin{defn}\label{D:smooth}
  Let $\M$ be a class of monomorphisms closed under isomorphisms and
  composition.
  \begin{enumerate}
  \item\label{D:smooth:1} We say that an object $A$ has \emph{\smooth
      $\M$-subobjects} provided that $\Sub_\M(A)$ is a dcpo with
    bottom (in particular, not a proper class) where the least element
    and directed joins are given by colimits of the corresponding
    diagrams of subobjects.

  \item\label{D:smooth:2} The class $\M$ is \emph{smooth} if every
    object of $\A$ has smooth $\M$-subobjects.
  \end{enumerate}
  Moreover, we say that a category has \emph{smooth monomorphisms} if
  the class of all monomorphisms is smooth.
\end{defn}
\begin{rem}\label{R:smooth}\mbox{ }
  \begin{enumerate}
  \item In more detail, let $D \subseteq \Sub_\M(A)$ be a directed set
    of subobjects represented by $m_i\colon A_i \monoto A$
    ($i \in D$). Then $D$ has a join $m\colon C \monoto A$ in
    $\Sub_\M(A)$. Moreover, consider the diagram of objects
    $(A_i)_{i \in D}$ with connecting morphisms
    $a_{i,j}\colon A_i \monoto A_j$ for $i \leq j$ in $D$ given by the
    unique factorizations witnessing $m_i \leq m_j$:
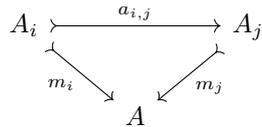
    \[
      \begin{tikzcd}
        A_i \arrow[>->]{rr}{a_{i,j}} \arrow[>->]{rd}[swap]{m_i} & & A_j
        \arrow[>->]{ld}{m_j} \\
        & A
      \end{tikzcd}
    \]
    (Note that $a_{i,j}$ need not lie in $\M$.)  Then for every
    $i \in D$ there exists a monomorphism $c_i\colon A_i \monoto C$ with
    $m \cdot c_i = m_i$, since $m_i \leq m$. The smoothness requirement
    is that these monomorphisms form a colimit cocone.
    
  \item Requiring that the least subobject in $\Sub_\M{A}$ is given by
    (the empty) colimit means that $\A$ has an initial object $0$ and
    the unique morphism $0 \monoto A$ 
    lies in $\M$.
  \item\label{R:smooth:3}
    If $\M$ is a smooth class, then $\A$ is $\M$-well-powered.
  \end{enumerate}
\end{rem}

Since the above notion of smoothness is new, we discuss examples at
length now.  Below we show that in a number of categories the
collection of all monomorphisms is smooth, as is the collection of all
strong monomorphisms (those having the diagonal fill-in property with
respect to epimorphisms).
%
%
We also present some counterexamples and
discuss other classes~$\M$.

Recall the concept of a \emph{locally finitely presentable}
(\emph{lfp}, for short) category (e.g.~\cite{ar}): it is a cocomplete
category $\A$ with a set of finitely presentable objects (i.e.~their
hom-functors preserve filtered colimits) whose closure under filtered
colimits is all of $\A$. Examples are $\Set$, $\Pos$ (posets and
monotone maps), $\Gra$ (graphs and homomorphisms) and all varieties of
finitary algebras such as monoids, vector spaces, rings, etc.

We say that $\A$ has a \emph{simple} initial object $0$ if all the
morphisms with domain $0$ are strong monomorphisms (equivalently, $0$
has no proper quotients).

\begin{example}\label{E:con}
  Both monomorphisms and strong monomorphisms are \smooth in every lfp
  category with a simple initial object $0$~\cite[Cor.~1.63]{ar}. This
  includes $\Set$, $\Pos$, $\Gra$, monoids and vector spaces. But not
  rings: in that category the initial object is $\Z$, the ring of
  integers, and there are non-monic ring homomorphisms with that
  domain (e.g.~$\Z \to 1$).
               
\end{example}
\begin{example}\label{E:dcpo}
  Let us consider the category $\DCPOb$ of dcpos with bottom and
  continuous maps between them, where a map is \emph{continuous} if it
  is monotone and preserves directed joins. 
  \begin{enumerate}
  \item In \autoref{S:dcpo} we prove that the class of all embeddings
    (\autoref{D:DCPOb-enriched}) is smooth.%
    \smnote{I think it's better to have the definition of embedding
      where it is (in the section where it is needed, and where it
      needs to be found easily, so let's make a forwards reference here.} 
    (These play a major
    role in Smyth and Plotkin's solution method for recursive domain
    equations~\cite{SmythPlotkin:82}.)  This example is one of several
    motivations for our move from the class of all monomorphisms to
    the more general situation of a class $\M$ in \autoref{D:smooth}.%

  \item\label{E:dcpo:2} In contrast, the class of all monomorphisms is non-\smooth in
    $\DCPOb$. For example, consider the dcpo $\N^\top$ of natural
    numbers with a top element $\top$. The subposets
    $C_n = \set{0, \ldots, n} \cup \set{\top}$, $n \in \N$, form an
    $\omega$-chain in $\DCPOb$. Its colimit is
    $\N^\top \cup\set{\infty}$ where $n < \infty < \top$ for all
    $n \in \N$.  The cocone of inclusion maps $C_n \subto \N^\top$
    consists of monomorphisms. However, the
    factorizing morphism from $\colim C_n$ to $\N^\top$ is not monic,
    as it merges $\infty$ and $\top$.

  \item The same example demonstrates that strong monomorphisms are
    not smooth in $\DCPOb$.
    
  \end{enumerate}
\end{example}
\begin{example}\label{E:cms}\mbox{ }
  \begin{enumerate}
  \item Let us consider the category $\MS$ of metric spaces with
    distances at most $1$ and \emph{non-expanding} maps
    $f \colon (X,d_X) \to (Y,d_Y)$ (that is
    $d_Y(f(x),f(y)) \leq d_X(x,y)$ for all $x,y, \in X$. Although this
    category is not lfp, both monomorphisms and strong monomorphisms
    form smooth classes. The proof for strong monomorphisms is easy
    since the strong (equivalently, extremal) subobjects of a metric
    space $A$ are represented by its subspaces (with the inherited
    metric). Given a directed set of subspaces $A_d \subseteq A$
    ($d \in D$) their join in $\Sub(A)$ the subspace
    $\bigcup_{d \in D} A_d$ and this is also the colimit of the
    corresponding diagram in $\MS$. The somewhat technical proof for
    monomorphisms is given in the appendix (\autoref{L:MS-smooth}).

  \item In the full subcategory $\CMS$ of $\MS$ given by all complete
    metric spaces monomorpisms are not smooth. This can be
    demonstrated as in \itemref{E:dcpo}{E:dcpo:2}: Let $\N^\top$
    be the metric space with distances $d(n,m)= |1/2^{-n} -1/2^{-m}|$
    and $d(n,\top)= 1/2^{-n}$, and consider the $\omega$-chain of
    spaces $C_n$ where $d(n,\top)=1$ and other distances are as in
    $\N^\top$.

  \item In contrast, strong monomorphisms are smooth in $\CMS$ (see
    \autoref{L:CMS-smooth}). 
  \end{enumerate}
\end{example}
The following equivalent formulation is often used in proofs.
\begin{proposition}\label{P:smooth}
  An object $A$ has \smooth $\M$-subobjects if and only if for every
  directed diagram $D$ of monomorphisms in $\A$ (not necessarily
  members of $\M$), and every cocone $m_i\colon A_i \monoto A$,
  $i\in D$, of $\M$-monomorphisms, the following holds:
  \begin{enumerate}
  \item\label{P:smooth:1} the diagram $D$ has a colimit, and
  \item\label{P:smooth:2} the factorizing morphism induced by the cocone 
    $(m_i)$ is again an $\M$-monomorphism.
  \end{enumerate}
\end{proposition}
\begin{proof}
  The `only if' direction is obvious. For the `if' direction, suppose
  we are given a directed set $D \subseteq \Sub(A)$ of $\M$-subobjects
  $m_i\colon A_i \monoto A$ for $i \in D$ as in~\autoref{R:smooth}.
  By~\autoref{P:smooth:1}, the ensuing directed diagram of
  monomorphisms $a_{i,j}\colon A_i \monoto A_j$ has a colimit
  $c_i\colon A_i \to C$, $i \in D$, and we will prove that this yields
  the join $\bigvee_{i \in D} m_i$. By~\autoref{P:smooth:2}, we have a
  unique $\M$-monomorphism $m\colon C \monoto A$ such that
  $m \cdot c_i = m_i$ for all $i \in D$.

  Now let $s\colon S \monoto A$ be any $\M$-subobject with $m_i \leq s$
  for all $i \in D$. That is, we have morphisms $s_i\colon A_i
  \to S$ with $s \cdot s_i = m_i$ for all $i \in D$. They
  form a cocone because for the monomorphism $a_{i,j}\colon A_i \monoto A_j$
  witnessing $m_i \leq m_j$ we have
  \[
    s \cdot s_j \cdot a_{i,j} = m_j \cdot a_{i,j} = m_i = s\cdot s_i,
  \]
  whence $s_j \cdot a_{i,j} = s_i$ since $s$ is monic. We therefore
  obtain a unique $t\colon C \to S$ with $t \cdot c_i = s_i$ for
  all $i \in D$. Consequently, we have
  \[
    s \cdot t \cdot c_i = s\cdot s_i = m_i = m \cdot c_i
    \qquad
    \text{for all $i \in D$.}
  \]
  Since the colimit injections $c_i$ form an epic family, we conclude
  that $s \cdot t = m$, which means that $m \leq s$ in $\Sub_\M(A)$, as
  desired.
\end{proof}
\begin{rem}\mbox{ }
  \begin{enumerate}
  \item Note that the conditions for $\M$ to be \smooth are a part of
    the conditions of Taylor's notion of a \emph{locally complete
      class of supports}~\cite[Def.~6.1. \& 6.3]{taylor3} (see
    also~\cite[Assumption~4.18]{Taylor21}).

  \item Smoothness previously appeared for joins and colimit of chains
    in lieu of directed sets~\cite{amm20}. That formulation is related
    to the list of conditions for a class of monomorphisms given by
    Trnkov\'a et al.~\cite{takr}. Note that a class $\M$ of
    monomorphisms containing the identities and closed under
    composition can be regarded as the subcategory of $\A$ given by
    all morphisms in $\M$. The list of conditions in op.~cit.~is
    equivalent to stating that the inclusion functor $\M \subto \A$
    creates colimits of chains. Requiring that the inclusion creates
    directed colimits implies that the class $\M$ is smooth. For the
    converse, we would need to add that for every directed diagram of
    $\M$-monomorphisms the colimit cocone consists of
    $\M$-monomorphisms.
  \end{enumerate}
\end{rem}

\section{The Initial Algebra Theorem}
\label{S:ini-thm}

We are now ready to prove the main result of this paper.

\begin{assumption}\label{A:ass}
  Throughout this section we assume that $\A$ is a category with a class $\M$ of
  monomorphisms containing all isomorphisms and closed under
  composition. We say that $F\colon \A \to \A$ \emph{preserves
    $\M$} if $m \in \M$ implies $Fm \in \M$.
\end{assumption}
\begin{defn}
  An \emph{$\M$-pre-fixed point} of $F$ is an algebra whose
  structure $m\colon FA \monoto A$ lies in $\M$. In the case where
  $\M$ consists of all monomorphisms we speak of a pre-fixed point.
%
\end{defn}
\begin{theorem}[Initial Algebra Theorem]\label{T:initial}
  Let  $m\colon FA \monoto A$ be an $\M$-pre-fixed point for an
  endofuctor preserving $\M$. If $A$ has
  \smooth $\M$-subobjects, then $F$ has an initial algebra which is an
  $\M$-subalgebra of~$(A,m)$.
\end{theorem}
\begin{proof}
  We have the following endomap
  \begin{equation}\label{eq:f}
    f\colon\Sub_\M(A)\to \Sub_\M(A)\quad\text{defined by}\quad 
    f\big(\begin{tikzcd}[cramped, sep = 20]
      B \arrow[>->]{r}{u} & A
    \end{tikzcd}%
    \big) = \big(%
    \begin{tikzcd}[cramped, sep = 20]
      FB \arrow[>->]{r}{Fu} &
      FA \arrow[>->]{r}{m} & A
    \end{tikzcd}%
    \big).
  \end{equation}
  It is clearly monotone. We are going to apply Pataraia Induction to
  it. We take the subset $S \subseteq \Sub_\M(A)$ of all
  $u\colon B \monoto A$ such that $u\leq f(u)$ via some recursive
  coalgebra $\beta\colon B \to FB$. More precisely,
  \[
    S = \set{u\colon B \monoto A :
      \text{$u = m\cdot Fu\cdot \beta$ for some recursive coalgebra $\beta\colon B \monoto FB$}}.
  \]
  Note that if $\beta$ exists for $u$, then it is unique. Moreover, $u \in S$ is a
  coalgebra-to-algebra morphism from $(B,\beta)$ to $(A,m)$.

  The least subobject $0 \monoto A$ is clearly contained in
  $S$.  Further, $S$ is closed under $f$ since $(FB, F\beta)$ is a recursive
  coalgebra by \autoref{P:recFappl}: for $u \in S$ we have%
  \[
    f(u) = m \cdot Fu = m \cdot F(m \cdot Fu \cdot \beta) = m \cdot F(f(u))
    \cdot F\beta.
  \]
  We continue with the verification that $S$ is closed under directed
  joins. Let $D\subseteq S$ be directed. Given $u\colon B_u \monoto A$ in $D$ we write
  $\beta_u\colon B_u \to FB_u$ for the recursive coalgebra witnessing
  $u \leq f(u)$. We show that these recursive coalgebras form a (then
  necessarily) directed diagram. To see this,
  we only need to prove that every morphism $h\colon B_u \monoto B_v$
  witnessing $u \leq v$ in $D$; i.e.~$v \cdot h = u$, is a coalgebra
  homomorphism. Consider the diagram below:
    \[
    \begin{tikzcd}
      B_u
      \arrow{r}{\beta_u}
      \arrow{d}[swap]{h}
      \arrow[shiftarr = {xshift=-20}]{dd}[swap]{u}
      &
      FB_u
      \arrow{d}{Fh}
      \arrow[shiftarr = {xshift=25}]{dd}{Fu}
      \\
      B_v
      \arrow{r}{\beta_v}
      \arrow{d}[swap]{v}
      &
      FB_v
      \arrow[>->]{d}{Fv}
      \\
      A
      &
      FA
      \arrow[>->]{l}[swap]{m}
    \end{tikzcd}
  \]
  Since the outside, the lower square and the left-hand and right-hand
  parts commute, we see that the upper square commutes when extended
  by the monomorphism $m \cdot Fv$. Thus it commutes,
  proving that $h$ is a coalgebra homomorphism.

  Now denote by $v\colon B \monoto A$ the join $\bigvee D$ in
  $\Sub_\M(A)$. Since $A$ has smooth subobjects, $B$ is the colimit of
  the diagram formed by the $B_u$, $u \in D$, in $\A$. Since the
  forgetful functor $\Coalg F \to \A$ creates colimits, we have a
  unique coalgebra structure $\beta\colon B \to FB$ such that the
  colimit injections are coalgebra homomorphisms; moreover $(B,\beta)$
  is colimit of the coalgebras $(B_u, \beta_u)$, $u \in D$. Thus,
  $(B,\beta)$ is recursive by \autoref{P:colimrec}. Moreover,
  $v\colon B \monoto A$ is the unique morphism induced by the cocone
  given by all $u\colon B_u \monoto A$ in $D$. Since every
  $u \in S$ is the unique coalgebra-to-algebra morphism from
  $(B_u,\beta_u)$ to $(A,m)$, we know from the proof of
  \autoref{P:colimrec} that $v$ is the unique coalgebra-to-algebra
  morphism from $(B,\beta)$ to~$(A,m)$. Thus, $v$
  lies in $S$.
  
  By \autoref{T:Pataraia}, $f$ has a least fixed point, and
  by~\autoref{C:PataInd}, $\mu f\in S$. Denote this subobject be
  $u\colon I \monoto A$.  Since $u \in S$, there is a recursive
  coalgebra $\iota\colon I \to FI$ such that $u = m \o Fu \o \iota$.
  But $u$ and $f(u) = m\cdot Fu$ represent the same subobject of $A$.
  So $\iota$ is an isomorphism.  Thus $(I,\iota^{-1})$ is an initial
  algebra by \autoref{R:iso-ini}.
\end{proof}
\begin{corollary}\label{C:initial}
  Let $\A$ be a category with a \smooth class  $\M$ of
  monomorphisms.  Then the
  following are equivalent for every endofunctor $F$ preserving $\M$:
  \begin{enumerate}
  \item\label{C:initial:1} an initial algebra exists,
  \item\label{C:initial:2} a fixed point exists,  
  \item\label{C:initial:3} an $\M$-pre-fixed point exists. 
  \end{enumerate}
  Moreover, if these hold, then $\mu F$ is an
  $\M$-subalgebra of every $\M$-pre-fixed point of $F$.
\end{corollary}
Indeed, Lambek's Lemma~\cite{lambek} tells us that \ref{C:initial:1}
implies \ref{C:initial:2}. Clearly, \ref{C:initial:2} implies
\ref{C:initial:3} since $\M$ contains all isomorphisms.
\autoref{T:initial} shows that that \ref{C:initial:3} implies
\ref{C:initial:1}, and it also yields our last statement.
\begin{corollary}
  Let $\A$ be an lfp category with a simple initial object. An
  endofunctor preserving monomorphisms has an initial algebra iff it
  has a pre-fixed point. 
\end{corollary}
\begin{example}\label{E:ass-necc}
  We present examples which show, \emph{inter alia}, that neither of
  the hypotheses in \autoref{T:initial} can be left out.  In each case
  $\M$ is the class of all monomorphisms.
  \begin{enumerate}
  \item The assumption that $0 \to A$ is monic. Let $\A$ be the
    variety of algebras $(A,u,c)$ with unary operation $u$ and a
    constant $c$. Its initial object is $(\N, s, 0)$ with
    $s(n) = n+1$, which is not simple. We present an endofunctor
    having no initial algebra
    even though it has a fixed point and preserves
    monomorphisms. Let $\Pow_0$ be the non-empty power-set functor. We
    obtain an analogous endofunctor
    $\bar\Pow_0$ on~$\A$ defined by
    $\bar\Pow_0(A,u,c) = (\Pow_0 A, \Pow_0u, \set{c})$. It clearly
    preserves monomorphisms, and the terminal object $1$ is a fixed
    point of $\bar \Pow_0$ (since $\Pow_0 1 \cong 1$). This is, up to
    isomorphism, the only fixed point. However, it is not
    $\mu \bar\Pow_0$ because given an algebra on $(A,u,c)$ with
    $u(x) \neq x$ for all $x \in A$, no $\bar\Pow_0$-algebra
    homomorphism exists from $1$ to $A$.

  \item The assumption that $\Sub(A)$ is a set in
    \itemref{D:smooth}{D:smooth:1}. Let $\Ord$ be the totally
    ordered class of all ordinals taken as a category. In the opposite
    category $\Ord^\opp$, all morphisms are monic, so every
    endofunctor preserves monomorphisms. For the functor $F$ on
    $\Ord^\opp$ given by $F(i) = i+1$,
    every object is a pre-fixed point, and there are no fixed points.
    For each object~$i$, $\Sub(i)$ has all the properties requested in
    \itemref{D:smooth}{D:smooth:1} except that it is a proper
    class.
    
  \item\label{E:ass-necc:3} Preservation of monomorphisms. Here we
    use the category $\Set \times \Set$ which satisfies all
    assumptions of \autoref{T:initial-2}. We define an endofunctor $F$
    by $F(X,Y) = (\emptyset, 1)$ if $X \neq \emptyset$ and $F(X,Y) =
    (\emptyset, \Pow Y)$ else. 
    It is defined on morphisms as expected, using $\Pow$ in the case
    where $X = \emptyset$.%
    \smnote{I vote for leaving this; it's still correct English and
      the type setting looks better this way (not with a line starting
      with $X$).}
    This functor has many pre-fixed points,
    e.g.~$F(1,1) = (\emptyset, 1) \monoto (1,1)$. But it has no fixed
    points (thus no initial algebra): first, $(\emptyset,Y)$ and
    $(\emptyset, \Pow Y)$ are never isomorphic, by Cantor's
    Theorem~\cite{Cantor91}. Second, if $X \neq \emptyset$, then there
    exists no morphism from $(X,Y)$ to $F(X,Y) = (\emptyset, 1)$.
  \end{enumerate}
\end{example}
\section{Initial Algebras in $\DCPOb$-enriched Categories}
\label{S:dcpo}

It follows from the seminal paper by Smyth and
Plotkin~\cite{SmythPlotkin:82} that every locally continuous functor
$F$ on a category $\A$ enriched over $\omega$-cpos (i.e.~partial
orders with a least element and joins of $\omega$-chains) has an
initial algebra $(\mu F, \ini)$ which is also a terminal coalgebra by
inverting its structure. Local continuity means that the corresponding
mappings $\A(A,B) \to \A(FA,FB)$ preserve (pointwise) directed joins. Here we assume
the weaker property that $F$ is locally monotone; for example, the
endofunctor assigning to a dcpo its ideal completion is locally monotone, whence preserves
embeddings, but not locally continuous. We apply \autoref{C:initial}
to derive that such an endofunctor has a pre-fixed point given by an
embedding iff it has an initial algebra (being also the terminal
coalgebra).
\begin{defn}\label{D:DCPOb-enriched}\mbox{ }
  \begin{enumerate}
  \item A category $\A$ is \emph{$\DCPOb$-enriched} provided that each
    hom-set is equipped with the structure of a dcpo with bottom, and
    composition preserves bottom and directed joins:%
    \smnote{Without explanation, it is at least not sufficient clear
      what `preserves bottom' means; it could mean
      $\bot \cdot \bot = \bot$.}  for every morphism $f$ and
    appropriate directed sets of morphisms $g_i$ ($i \in D$) we have
    \begin{equation}\label{eq:dcpo-enriched}
      \textstyle
      f \cdot \bot = \bot, \quad \bot \cdot f = \bot, \quad
      f \cdot \bigvee_{i\in D} g_i = \bigvee_{i\in D} f \cdot g_i,\quad
      \big(\bigvee_{i \in D} g_i\big) \cdot f = \bigvee_{i\in D} g_i
      \cdot f.
    \end{equation}
  \item A functor on $\A$ is \emph{locally monotone} if its
    restrictions $\A(A,B) \to \A(FA,FB)$ to the hom-sets are monotone.
    
  \item A morphism $e\colon A \to B$ is called an \emph{embedding} if
    there exists a morphism $\ehat\colon B \to A$ such that
    $\ehat \cdot e = \id_A$ and $e \cdot \ehat \sqsubseteq \id_B$.
  \end{enumerate}
\end{defn}
\noindent
It is easy to see that the morphism $\ehat$ is unique for $e$; it
is called its \emph{projection}. 

The following result is a slight variation of a result by Smyth and
Plotkin for $\omega$-cpos~\cite{SmythPlotkin:82}.  We include the
proof in the appendix for the convenience of the reader.
\begin{theorem}\label{Basic-Lemma-D}
  Let $D$ be a directed diagram of embeddings in a $\DCPOb$-enriched
  category. For every cocone $(c_i\colon D_i \to C)$ of $D$, the
  following are equivalent:
  \begin{enumerate}
  \item\label{Basic-Lemma:1D} The cocone $(c_i)$ is a colimit.
  \item\label{Basic-Lemma:2D} Each $c_i$ is an embedding, the composites
    $c_i \cdot \widehat{c}_i$ form a directed set in $\A(C,C)$, and
    \begin{equation}\label{eq-supD}
      \textstyle
      \bigsqcup_{i} c_i \cdot \widehat{c}_i = \id_C.
    \end{equation}
  \end{enumerate}
\end{theorem}
\begin{rem}
  A $\DCPOb$-enriched category $\A$ is $\M$-well-powered for the class
  $\M$ of all embeddings. The reason is that, given an object $A$, a
  subobject represented by an embedding $e\colon S \to A$ is
  determined by the endomorphism $e\cdot \ehat$ on $A$. Indeed, let
  $f\colon T \to A$ be an embedding with
  $e\cdot \ehat = f \cdot \widehat f$. Then
  $e = e \cdot \ehat \cdot e = f \cdot \widehat f \cdot e$.
  Therefore, $e \leq f$ in $\Sub_\M(A)$. By symmetry $f \leq e$. Since
  $\A(A,A)$ is a set, $\M$-well-poweredness follows.
\end{rem}
\begin{theorem}\label{T:dcpos-smooth}
  Let $\A$ be a $\DCPOb$-enriched category with directed
  colimits. Then the class of all embeddings is \smooth.
\end{theorem}
\noindent
The proof is presented in \autoref{S:dcpos-smooth}.
%
%
\begin{corollary}\label{C:dcpo}
  Let $\A$ be a $\DCPOb$-enriched category with directed colimits.
  For a locally monotone endofunctor $F$ the following are equivalent:
  \begin{enumerate}
  \item\label{C:dcpo:1} an initial algebra exists,
  \item\label{C:dcpo:2} a terminal coalgebra exists,
  \item\label{C:dcpo:3} a fixed point exists.
  \end{enumerate}
  Moreover, if $(\mu F, \ini)$ is an initial algebra, then $(\mu F,
  \ini^{-1})$ is a terminal coalgebra.
\end{corollary}
\renewcommand{\itemautorefname}{Item}
\autoref{C:dcpo:3} can be strengthened to state existence of a
pre-fixed point carried by an embedding.
\renewcommand{\itemautorefname}{item}
\begin{proof}
  The dual category $\A^\opp$ is $\DCPOb$-enriched w.r.t.~the same
  order on hom-sets. But the embeddings in $\A^\opp$ are precisely the
  projections in $\A$. Every locally monotone endofunctor $F$ on $\A$
  clearly preserves embeddings and  projections. Thus, the dual
  functor $F^\opp$ on $\A^\opp$ preserves embeddings.
  Now \ref{C:dcpo:1} $\Leftrightarrow$ \ref{C:dcpo:3} follows from an
  application of \autoref{C:initial} to $\A$ and $F$, and
  \mbox{\ref{C:dcpo:2} $\Leftrightarrow$ \ref{C:dcpo:3}} is an
  application to $\A^\opp$ and $F^\opp$. In each case the class $\M$
  consists of all embeddings in $\A$ and $\A^\opp$, respectively.
  
  Finally, we prove that the initial algebra and terminal coalgebra
  coincide. Let $\ini\colon F I \to I$ be an initial algebra. Then we
  know that a terminal coalgebra $\ter\colon T \to FT$
  exists. Moreover, from the last statement in \autoref{C:initial}
  applied to $F$ and its fixed point $(T,\ter^{-1})$ we see that the
  unique $F$-algebra homomorphism $e\colon (I,\ini) \to
  (T,\ter^{-1})$ is an embedding. Another application of \autoref{C:initial}
  to $F^\opp$ and its fixed point $(I,\ini^{-1})$ yields
  that the unique~$F^\opp$-algebra homomorphism $f\colon\mu F^\opp =
  (T,\ter) \to (I,\ini^{-1})$ is an embedding in $\A^\opp$. This
  means that this an $F$-coalgebra homomorphism $f\colon (I,\ini^{-1}) \to
  (T,\ter)$ which is a projection in $\A$. By the universal properties
  of $(I,\ini)$ and $(T,\ter)$, $e = f$, and this morphism is both an
  embedding and a projections, whence an isomorphism.  
\end{proof}
The requirement of local monotonicity of $F$ can be weakened: the
theorem holds for any endofunctor $F$ which fulfils
$Ff \sqsubseteq \id_{FA}$ whenever $f \sqsubseteq \id_{A}$. Indeed, a
functor satisfying that property preserves embedding-projection pairs;
in categories with split idempotents the converse holds,
too~\cite[Obs.~6.6.5]{amvbook}.

We close this section with a proposition on locally
monotone functors which gives a version of a result for
$\omega$-cpo-enriched categories proved by
Freyd~\cite{freyd:Cambridge} for locally continuous functors. He used
Kleene's Theorem in lieu of Pataraia's.%
\begin{proposition}
  Let $\A$ be a $\DCPOb$-enriched category. If a locally monotone
  functor~$F$%
  \smnote{No `endo' here; I'd like $F$ to stay on that line and not
    start the next line.}  has an initial algebra $(\mu F, \ini)$,
  then $(\mu F, \ini^{-1})$ is a terminal coalgebra.
\end{proposition}
\noindent
We shall see in the proof that it is enough to assume that composition
is left-strict: $\bot \cdot f = \bot$ holds for every morphism $f$ of
$\A$ (but $f \cdot \bot = \bot$ in~\eqref{eq:dcpo-enriched} need not
hold). This holds in categories typically used in semantics of
programming languages, such as the category of dcpos with bottom and
(non-strict) continuous maps, where composition is not (right-)
strict.
\begin{proof}
  Let $\ini\colon F I \to I$ be an initial algebra. For every
  coalgebra $\alpha\colon A \to FA$, we prove that a unique
  homomorphism into $(I, \ini^{-1})$ exists.
  \begin{enumerate}
  \item Existence. The endomap $g$ on $\A(A,I)$ given by
    $h \mapsto \ini \cdot Fh \cdot \alpha$ is monotone since $F$ is
    locally monotone. Hence, it has a least fixed point
    $h\colon A \to I$ with $\ini^{-1}\cdot h = Fh \cdot \alpha$ by
    Pataraia's \autoref{T:Pataraia}. This is a coalgebra homomorphism.

  \item Uniqueness. First notice that for $\A(I,I)$ we have an the
    analogous endomap $f$ given by $k \mapsto \ini \cdot Fk \cdot
    \ini^{-1}$. Since $I$ is initial, the only fixed point of $f$ is
    $k = \id_I$. Thus $\id_I = \mu f$. Now
    suppose that $h'\colon (A, \alpha) \to (I,\ini^{-1})$ is any
    coalgebra homomorphism. We know that
    $\A(h',I)\colon \A(I,I) \to \A(A,I)$ is a strict continuous map;
    strictness follows from left-strict\-ness of composition:
    $\bot_{I,I} \cdot h' = \bot_{A,I}$. We now show that
    $g \cdot \A(h',I) = \A(h',I) \cdot f$. Indeed, unfolding
    the definitions, we have for every $k\colon I \to I$:
    \begin{align*}
      g \cdot \A(h',I)(k)
      & =
      g(k \cdot h')
      =
      \ini \cdot F(k\cdot h') \cdot \alpha
      =
      \ini \cdot Fk \cdot Fh' \cdot \alpha
      =
      \ini \cdot Fk \cdot \ini^{-1} \cdot h' \\
      & =
      f(k) \cdot h' =  \A(h',I) (f(k)).
    \end{align*}    
    Therefore, by \autoref{L:mu-pres}, $\A(h',I)(\mu f) = \mu
    g$, which means that $h' = \id_I \cdot h' = h$.\qedhere%
  \end{enumerate}
\end{proof}

We leave as an open problem to find an endofunctor on $\DCPOb$ which
has a fixed point but not an initial algebra.

\section{The Initial-Algebra Chain}
\label{S:ini-chain}

The proof of \autoref{T:initial}, relying on Pataraia's
\autoref{T:Pataraia}, is constructive. However, if one admits
non-constructive reasoning and ordinals, then we can add another
equivalent characterization to \autoref{C:initial} in terms of the
convergence of the initial-algebra chain, which we now recall.
\begin{rem}\label{R:chain}\mbox{ }
  \begin{enumerate}
  \item Recall that an ordinal $i$ is the (linearly ordered) set of
    all ordinals smaller than $i$. As such it is also a category. 

  \item\label{R:chain:2} By an \emph{$i$-chain} in a category $\C$ is
    meant a functor $C\colon i \to \C$. It consists of objects $C_j$
    for all ordinals $j < i$ and (connecting) morphisms
    $c_{j,j'}\colon C_j \to C_{j'}$ for all pairs $j \leq j' < i$.
    Analogously, an \emph{$\Ord$-chain} in $\C$ is a functor from the
    totally ordered class $\Ord$ of all ordinals to $\C$.  In both
    cases we will speak of a (transfinite) \emph{chain}
    whenever confusion is unlikely.
    
  \item\label{R:chain:3} A category $\C$ \emph{has colimits of chains} if for every
    ordinal $i$ a colimit of every $i$-chain exists in $\C$. (This
    does \emph{not} include $\Ord$-chains.) In particular, $\C$ has an
    initial object since the ordinal $0$ is the empty set.

%
%
  \end{enumerate}
\end{rem}
\removeThmBraces
\begin{defn}[{\cite{A74}}]\label{D:chain}
  Let $\A$ be a category with colimits of chains. For an endofunctor
  $F$ we define the \emph{initial-algebra chain}
  $W\colon \Ord \to \A$. Its objects are denoted by $W_i$ and its
  connecting morphisms by $w_{ij}\colon W_i \to W_j$,
  $i \leq j \in \Ord$. They are defined by transfinite recursion as follows
  \[
    \begin{array}{@{}l}
      W_0 = 0,\quad
      \text{$W_{j+1} = FW_j$ for all ordinals $j$},\quad
      \text{$\textstyle W_j =  \colim_{i<j} W_i$ for all limit ordinals $j$},
      \\
      \text{$w_{0,1}\colon 0\to  W_1$ is unique}, \qquad
      w_{j+1,k+1} = F w_{j,k} \colon FW_j \to FW_k, \\
      \text{$w_{i,j}$ ($i<j$) is the colimit cocone for limit ordinals
        $j$}
    \end{array}
  \]
\end{defn}
\resetCurThmBraces
\begin{rem}\label{R:ini-chain}\mbox{ }
  \begin{enumerate}
  \item There exists, up to natural isomorphism, precisely one
    $\Ord$-chain satisfying the above equations. For example,
    $w_{\omega, \omega+1}\colon W_\omega \to FW_\omega$ is determined
    by the universal property of
    $W_\omega = \colim_{n< \omega} W_n = \colim_{n<\omega} W_{n+1}$ as
    the unique morphism with
    $w_{\omega,\omega+1} \cdot w_{n+1,\omega} = w_{n+1,\omega+1} =
    Fw_{n,\omega}$ for every $n < \omega$.
  \item\label{R:ini-chain:2}
    Every algebra $\alpha\colon FA \to A$ induces a
    canonical cocone $\alpha_i \colon W_i \to A$ ($i \in \Ord$) 
    on the initial-algebra chain; it is the unique cocone with
    \(
    \alpha_{i+1} = (W_{i+1} = FW_i \xra{F \alpha_i} FA \xra{\alpha} A)
    \)
    for all ordinals $i$. This is easy to see using transfinite induction.
  \end{enumerate}
\end{rem}
\begin{defn}\label{D-initial-chain}
  We say that the initial-algebra chain of a functor $F$
  \emph{converges in $\lambda$ steps} if $w_{\lambda,\lambda+1}$ is an
  isomorphism, and we simply say that it \emph{converges}, if it converges
  in $\lambda$ steps for some ordinal $\lambda$. 
\end{defn}

If $w_{i,i+1}$ is an isomorphism, then so is $w_{i, j}$, for all
$j > \lambda$. This is easy to prove by transfinite induction.

Convergence of the initial-algebra chain yields
an initial algebra~\cite{A74}. We obtain this as a consequence
of results from \autoref{S:recoalg} on recursive coalgebras: 
\begin{theorem}\label{T:initial-chain} 
  Let $\A$ be a category with colimits of chains. If the
  initial-algebra chain of an endofunctor $F$ converges in $\lambda$
  steps, then $W_\lambda$ is the initial algebra with the algebra
  structure
  \(
  w^{-1}_{\lambda,\lambda+1}\colon  FW_\lambda \to W_\lambda. 
  \)
\end{theorem}
\begin{proof}
  An easy transfinite induction shows that every coalgebra
  $w_{i,i+1}\colon W_i \to FW_i$ is recursive: the
  coalgebra $0 \to F0$ is trivially recursive, for the isolated step
  use \autoref{P:recFappl}, and~\autoref{P:colimrec} yields the limit
  step. If $w_{\lambda,\lambda+1}$ is an isomorphism, then
  $(W_\lambda, w_{\lambda,\lambda+1}^{-1})$ is the initial algebra
  by~\autoref{R:iso-ini}.
\end{proof}
The existence of an $\M$-pre-fixed point implies that the
initial-algebra chain converges. The proof below is somewhat similar
to the proof of~\autoref{T:initial}. The difference is that one only
uses the recursive coalgebras $W_i\to FW_i$ in the initial-algebra
chain and applies Zermelo's \autoref{T:Zermelo} in lieu of Pataraia's
Theorem. For this we work again under \autoref{A:ass}.

\begin{theorem}\label{T:initial-2}
  Let $F$ preserve $\M$ and $m\colon FA \monoto A$ be an
  $\M$-pre-fixed point. If $A$ has \smooth $\M$-subobjects, then the
  initial-algebra chain for $F$ converges.
\end{theorem}
\begin{proof}
  Again, we use the monotone endomap
  $f\colon \Sub_\M(A) \to \Sub_\M(A)$
  in~\eqref{eq:f}. \autoref{T:Zermelo} applies since $\Sub_\M(A)$ is a
  dcpo by assumption, and therefore it is a chain-complete poset.
  Thus, $f$ has the least fixed point $\mu f = f^i(\bot)$ for some
  ordinal $i$.  The cocone $m_j\colon W_j \to A$ of
  \itemref{R:ini-chain}{R:ini-chain:2} satisfies $m_j = f^j(\bot)$
  for all $j \in \Ord$. This is easily verified by transfinite
  induction. Hence, from $f(f^i(\bot)) = f^i(\bot)$ we conclude that
  $m_i$ and~$m_{i+1}$ represent the same subobject of $A$. Since
  $m_i = m_{i+1} \cdot w_{i,i+1}$, it follows that $w_{i,i+1}$ is
  invertible, which means that the initial-algebra chain converges.
\end{proof}
We now obtain the original
initial-algebra theorem by Trnkov\'a et al.~\cite{takr}:
\begin{corollary}\label{C:initial-2}
  Let $\A$ be a category with colimits of chains and
  \smnote{The referee is right; either we add existence of colimits of
    chains or we should add `the initial algebra chain exists' in
    point~(1); I prefer the latter. Note that having a smooth class implies
    $\M$-well-poweredness (\itemref{R:smooth}{R:smooth:3}); so we
    need not assume it.}
  with a \smooth class $\M$ of monomorphisms. Then the following are
  equivalent for an endofunctor $F$ preserving $\M$:
  \begin{enumerate}
  \item\label{C:initial-2:1} the initial-algebra chain converges,
  \item\label{C:initial-2:2} an initial algebra exists,
  \item a fixed point exists,
  \item\label{C:initial-2:4} an $\M$-pre-fixed point exists.
  \end{enumerate}
  Moreover, if these hold, then $\mu F$ is an
  $\M$-subalgebra of every $\M$-pre-fixed point of $F$.
\end{corollary}
Indeed, \ref{C:initial-2:4} implies
\ref{C:initial-2:1} by \autoref{T:initial-2}, and~\ref{C:initial-2:1}
implies~\ref{C:initial-2:2} is shown as in \autoref{T:initial-chain}. The remaining
implications are as for \autoref{C:initial}.
\begin{rem}
  Note that in lieu of assuming that $\A$ has colimits of \emph{all}
  chains, it suffices that the initial-algebra chain exists (i.e.~the
  colimits in \autoref{D:chain} exist). This weaker condition enables
  more applications, e.g.~the category of relations with $\M$ the
  class of injective maps and functors $F$ which are lifted from
  $\Set$.
\end{rem}
\begin{rem}\label{R:set}
  For a set functor $F$ no side condition is needed: if $F$ has a
  pre-fixed point, then it has an initial algebra. This is clear if
  $F\emptyset=\emptyset$. If not, there is a set functor $G$ with
  $G\emptyset \neq \emptyset$ which preserves monomorphisms and agrees
  with $F$ on all nonempty sets and maps~\cite{trnkova71}. Since every
  pre-fixed point of $F$ must be nonempty, it is also a pre-fixed
  point of~$G$. Hence~$G$ has an initial algebra, which clearly is an
  initial algebra for $F$, too.
\end{rem}
\begin{corollary}\label{C:converges}
  An endofunctor on one of the categories $\Set$, $\Pfn$, or $\KVec$
  has an initial algebra iff it has a pre-fixed point.
\end{corollary}
\begin{proof}
  For $\Set$, use \autoref{R:set}. 
  For $\Pfn$ and $\KVec$, apply
  \autoref{C:initial-2} with~$\M$ the class of all monomorphisms
  (which are split and therefore preserved by every endofunctor).
\end{proof}

%
%
\bibliographystyle{plainurl}
\bibliography{refs}

%
%
\clearpage
\appendix
\section{Further Technical Details}

\subsection{Details for \autoref{E:cms}}

\begin{lemma}\label{L:MS-smooth}
  Monomorphisms are \smooth in $\MS$.
\end{lemma}
\begin{proof}
  Fix a space $(A,d)$, and consider a directed set $D$ of subobjects
  $m_i\colon (A_i,d_i) \monoto (A,d)$ ($i \in D$) with monomorphisms
  $a_{i,j}\colon A_i \monoto A_j$ witnessing $i \leq j$ in $D$. Let
  $\Abar = \bigcup_{i \in D} m_i[A_i]$, and let
  $\dbar\colon\Abar\to[0,1]$ be defined as follows:
  \begin{equation}\label{mssmooth}
    \dbar(x,y) = \inf \set{d_i(x',y') : \text{$i \in D,
        x',y'\in A_i$, $m_i(x') = x$ and $m_i(y') = y$}}.
  \end{equation}
  We show that $\dbar$ is a metric. It is clearly symmetric and
  fulfils $D'(x,x) = 0$. We verify that distinct points $x, y$ in
  $\Abar$ have non-zero distance.  For each $i$, $x'$, and $y'$ as in
  \eqref{mssmooth}, $d_{i}(x',y') \geq d(x,y)$, since $m_i$ is
  non-expanding. Thus $\dbar(x,y) \geq d(x,y) > 0$.
  
  Finally, we verify that $\dbar$ satisfies the triangle
  inequality. To this end it suffices to show that for all $x, y, z
  \in \Abar$ and every $\eps > 0$ we have $\dbar(x,z) \leq \dbar(x,y)
  + \dbar(y,z) + \eps$.
  
  Let $x,y,z\in\Abar$ and fix $\varepsilon > 0$. We can choose $i$,
  $x'$, $y'$ as in \eqref{mssmooth} such that
  $d_{i}(x',y') < \dbar(x,y) + \varepsilon/2$. Analogously, let $j$,
  $y''$, $z''$ be such that
  $d_{j}(y'',z'') < \dbar(x,y) + \varepsilon/2$.  Since the collection
  $m_i[A_i]$ is directed, we can assume $i \leq j$ in $D$. Using that
  the connecting map $a_{i,j}$ is non-expanding we obtain
  $d_{j}(a_{i,j}(x'),a_{i,j}(y')) < \dbar(x,y) + \varepsilon/2$.
  Since $m_j$ is injective, $a_{i,j}(y') = y''$.  Let
  $x'' = a_{i,j}(x')$, and note that $m_j(x'') = x$.  By the triangle
  inequality in $A_j$,
  \[
    d_{j}(x'', z'') \leq d_{j}(x'', y'') + d_{j}(y'', z'')
    <
    \dbar(x,y) + \dbar(y,z) + \varepsilon.
  \]
  It follows that
  $\dbar(x,z) \leq \dbar(x,y) + \dbar(y,z) + \varepsilon$, as
  desired.
  
  It is obvious that the inclusion $m\colon \Abar\monoto A$ is
  non-expanding. It is also easy to check that $(\Abar,\dbar)$ is the
  join in $\Sub(A,d)$ of the directed diagram corresponding to given
  directed set $D$.

  Finally, for every $i \in D$, we have the codomain restriction
  $m_i'\colon A_i\monoto \Abar$ of $m_i$, which is non-expanding. We
  verify that the family of all $m'_i$ ($i \in D$) forms a colimit
  cocone. It clearly is a cocone. Consider any cocone
  $f_i\colon(A_i,d_i)\to (A^*,d^*)$, $i \in D$. Clearly, the union $B$
  is the colimit in $\Set$. Therefore, we have a unique map
  $f\colon \Abar \to A^*$ such that $f_i = f \cdot m_i'$ for all
  $i \in D$. This is given by $f(x) = f_i(x)$ whenever
  $x\in m_i[A_i]$. We check that $f$ is non-expanding, and this will
  conclude our verification. Let $x,y\in \Abar$, and choose $i,x', y'$
  as in~\eqref{mssmooth}. Since $f_i$ is non-expanding,
  \[
    d^*(f(x),f(y))
    =
    d^*(f_i(x'), f_i(y'))
    \leq
    d_i(x',y')
    \leq
    \dbar(x,y).\tag*{\qedhere}
  \]
\end{proof}
\begin{lemma}\label{L:CMS-smooth}
  Strong monomorphisms are smooth in $\CMS$.
\end{lemma}
\begin{proof}
  Fix a complete metric space $(A,d)$, and consider a directed set $D$
  of closed subspaces $A_i \subto A$. Their join $B \subto A$ is the
  closure of their union
  \[
     B = \overline{\bigcup_{i < \lambda} A_i}.
   \]
   We know from \autoref{L:MS-smooth} that the union is the colimit of
   the directed diagram corresponding to $D$ in $\MS$. Moreover, the
   colimit of a diagram in $\CMS$ is given by forming the Cauchy
   completion of the colimit of that diagram in $\MS$. (This follows
   from the fact that $\CMS$ is a reflective subcategory of $\MS$ with
   Cauchy completions as reflections.) Since $B$ is complete and
   $\bigcup_{i \in D} A_i$ is dense in it, $B$ is the Cauchy
   completion of that union, whence it is desired colimit in $\CMS$.
\end{proof}

\subsection{Proof of \autoref{Basic-Lemma-D}}
\begin{proof}
  \ref{Basic-Lemma:1D}~$\Rightarrow$ \ref{Basic-Lemma:2D}:
  Let $D$ have objects $D_i$ and connecting morphisms $e_{i,j}\colon
  D_i \to D_j$. Write~$\ehat_{i,j}$ for the projection of
  $e_{i,j}$.  We verify that for $i \leq j \leq k$,
  $\ehat_{i,k} = \ehat_{i,j} \cdot e_{j,k}$.  In fact, $\ehat_{i,j} $ is
  unique with $\ehat_{i,j} \cdot e_{i,j} = \id_{D_i}$ and
  $ e_{i,j} \cdot \ehat_{i,j} \sqsubseteq \id_{D_j}$.  But
  $\ehat_{i,k} \cdot e_{j,k}$ also has these properties, since
  \[
    \begin{aligned}
      (\ehat_{i,k} \cdot e_{j,k}) \cdot e_{i,j}
      & = \ehat_{i,k} \cdot e_{i,k} \\
      & = \id_{D_i},
    \end{aligned}
    \quad\qquad\text{and}\quad\qquad
    \begin{aligned}
      e_{i,j}\cdot  (\ehat_{i,k} \cdot e_{j,k})
      & = e_{i,j} \cdot \ehat_{i,j} \cdot \ehat_{j,k} \cdot e_{j,k}\\
      & = e_{i,j} \cdot \ehat_{i,j} \cdot \id_{D_k}\\
      & \sqsubseteq \id_{D_k}.
    \end{aligned}
  \]
  This shows that indeed $\ehat_{i,k} = \ehat_{i,j} \cdot e_{j,k}$ for
  $i\leq j\leq k$.
  
  For each $i$ form the subdiagram $D^{i}$ of all $D_j$ for $j\geq i$,
  with connecting maps $e_{j,k}$ for $i\leq j \leq k$ inherited from
  $D$.  Since $I$ is directed, the colimit of $D^{i}$ is
  $(c_j)_{j\geq i}$.  Our observation at the outset shows that we have
  a cocone of $D^i$:
  \[
    \takeout{
    \begin{tikzcd}[column sep = 15]
      D_j \ar{rr}{e_{j,k}}
      \ar{rd}[swap]{\ehat_{i,j}}
      &&
      D_k
      \ar{ld}{\ehat_{i,k}}
      \\
      &
      D_i
    \end{tikzcd}
  }
  \ehat_{i,j} = \big(D_j \xra{e_{j,k}} D_k \xra{\hat e_{i,k}} D_i\big).
  \]
  Thus, there is a unique factorization $\widehat{c}_i\colon C\to D_i$
  through the colimit cocone:
  \begin{equation}\label{DBL1}
    \ehat_{i,j} = \widehat{c}_{i} \cdot c_j  \quad \text{for $j\geq i$}.
  \end{equation}
  In particular, for $i = j$, we see that
  $\widehat{c}_i \cdot c_i = \id_{D_i}$.  We will verify below the
  equation $\bigsqcup_j c_j \cdot \chat_j = \id_C$,
  and this of course implies that
  $c_i \cdot \widehat{c}_i \sqsubseteq \id_C$.  (This justifies our use of
  the projection notation~$\widehat{c}_j$ and shows the first point in 
  \autoref{Basic-Lemma:2D}, that $c_i$ is an embedding.)
  
  Next, we show that 
  for each $j$, the morphisms $\chat_i$ for $i\leq j$ form a cone of $D^j$:
  \[
    \takeout{
    \begin{tikzcd}[column sep = 15]
      & C
      \ar{ld}[swap]{\chat_j}\ar{rd}{\chat_i}
      \\
      D_j \ar{rr}{\ehat_{i,j}} & & D_i  
    \end{tikzcd}
    }
    \chat_i = \big(C \xra{\hat c_j} D_j \xra{\hat e_{i,j}} D_i\big).
  \] 
  Indeed, the colimit cocone $(c_k)_{k\geq j}$ is collectively epic, so
  we need only establish this after precomposing with each $c_k$. We
  apply \eqref{DBL1} twice to obtain:
  \(
    \chat_{i} \cdot c_{k}
    =
    \ehat_{i,k}
    =
    \ehat_{i,j} \cdot \ehat_{j,k}
    =
    \ehat_{i,j} \cdot  \chat_{j} \cdot c_k.
  \)
  
  We are ready to argue for~\ref{Basic-Lemma:2D}.
  The maps $c_i\cdot \chat_i$ form a directed subset of $\A(C,C)$
  because for $i \leq j$,%
  \(
    c_i\cdot \chat_i
    =
    (c_j\cdot e_{i,j}) \cdot (\ehat_{i,j} \cdot \chat_j)
    \sqsubseteq
    c_j\cdot \chat_j.
  \)
  Thus, $\bigsqcup_{j} c_j\cdot \widehat{c}_j$ exists.
  We use that the family~$(c_i)$ is collectively epic, and verify that
  \(
    \textstyle
    \bigsqcup_j c_j \cdot \chat_j\cdot c_i = c_i \text{ for every $i$}.
  \)
  Fix $i$, and consider the join above.  Since it is over a directed
  set, we need only consider
  $\bigsqcup_{j\geq i} c_j \cdot \chat_j\cdot c_i $. In addition, for
  $j\geq i$ we obtain
  \(
  c_j \cdot \chat_j \cdot c_i
  =
  c_j \cdot \chat_j\cdot (c_j \cdot e_{i,j})
  =
  c_j \cdot  e_{i,j}
  =
  c_i.
  \)
  
  \smallskip
  \noindent
  \ref{Basic-Lemma:2D}~$\Rightarrow$~\ref{Basic-Lemma:1D}: Let
  $(b_i\colon D_i\to B)$ be a cocone.  For all $i \leq j$ we have
  $b_i = b_j \cdot e_{i,j}$.  We also have a cocone $(c_i)$, and so
  $c_i = c_j \cdot e_{i,j}$, and thus $\chat_i = \ehat_{i,j} \cdot \chat_j$.
  From this we have
  \[
    b_i\cdot \chat_i
    =
    (b_j \cdot e_{i,j}) \cdot  (\ehat_{i,j} \cdot \chat_j)
    \sqsubseteq
    b_j \cdot e_j.
  \] 
  Thus, the following join exists in $\A(C,B)$:
  \(
    \textstyle
    b = \bigsqcup_{j} b_j \cdot \chat_j.
  \)
  To prove that $b_i = b \cdot c_i$ for all $i$, we fix one $i$ and
  consider the join above with $j\geq i$:
  \[
    b_j \cdot \chat_j \cdot c_i
    =
    (b_j \cdot \chat_j )\cdot ( c_j \cdot e_{i,j})
    =
    b_j \cdot e_{i,j}
    =
    b_i.
  \]
  Thus
  \( b \cdot c_i = \bigsqcup_{j\geq i} (b_j\cdot \chat_j \cdot c_i) =
  \bigsqcup_{j\geq i} b_i = b_i.\) This shows that $b$ is the desired
  factorization of $(b_i)$. For its uniqueness, let $b'\colon C\to B$
  be a morphism with $b' \cdot c_i = b_i$ for all $i$.  Since
  $\bigsqcup_{i} c_i \cdot \chat_i = \id_C$, we have
  \(\textstyle b' = b' \cdot \big( \bigsqcup_{i} c_i \cdot
  \chat_i\big) = \bigsqcup_{i} b' \cdot c_i \cdot \chat_i =
  \bigsqcup_{i} b_i \cdot \chat_i = b.  \) This completes the proof.
\end{proof}

\subsection{Proof of~\autoref{T:dcpos-smooth}}
\label{S:dcpos-smooth}

\begin{proof}
We use \autoref{P:smooth}.   Fix an object $A$ in $\A$.
Let $\E$ be the class of embeddings, so that 
 $\Sub_{\E}(A)$ denotes the  poset of subobjects of $A$ represented by embeddings.    
Let $D$ be a directed diagram of monomorphisms in $\A$,  not necessarily embeddings,
and let 
  $m_i\colon A_i \monoto A$, $i\in D$, be a cocone of morphisms in $\Sub_{\E}(A)$.
By hypothesis, $D$ has a colimit cocone, say  $c_i\colon A_i \to B$.
We have a unique  morphism $m\colon B \to A$ such that for all $i$, 
$m_i =m \o c_i $.
Our task is to show that $m$ is an embedding,
and that $m = \bigsqcup_{i\in A} m_i$ in $\Sub_{\E}(A)$.

For $i\leq j$ in $A$, we have a morphism $e_{i,j}$ such that
$m_i = m_j \cdot e_{i,j}$.  Let us verify that $e_{i,j}$ is an
embedding and that $\ehat_{i,j}= \mhat_i\cdot m_j$.  To see this, we
use the characterization of projections. First,
$ (\mhat_i\cdot m_j) \cdot e_{i,j} = \mhat_i\cdot m_i = \id$. Second,
we verify $e_{i,j} \cdot (\mhat_i \cdot m_j) \sqsubseteq \id$:
\begin{align*}
  e_{i,j} \cdot  (\mhat_i\cdot m_j)
  & = (\mhat_j\cdot m_j)\cdot  e_{i,j} \cdot  \mhat_i\cdot m_j
  & \text{since  $\mhat_j\cdot m_j = \id$}\\
  & = \mhat_j\cdot m_i \cdot \mhat_i \cdot m_j
  &
  \text{since $m_i = e_{i,j}\o m_j$}
  \\
  & \sqsubseteq  \mhat_j\cdot   m_j
  & \text{since  $m_i \cdot \mhat_i= \id$}\\ 
  & = \id.
\end{align*}
We next show that for $i \leq j$, $c_i\o \mhat_i \sqsubseteq c_j \o \mhat_j$.
Once this is done, we put $\mhat = \bigsqcup_i c_i\o\mhat_i$ and show that it is a projection for $m$.
We thus calculate:
\begin{align*}
  c_i\o \mhat_i
  & = c_j \o e_{i,j} \o \mhat_i \\
  & = c_j \o e_{i,j}  \o \!\!\widehat{~m_j \cdot e_{i,j}}\\
  & = c_j  \o e_{i,j} \o \ehat_{i,j}\o \mhat_j \\
  & \sqsubseteq c_j \o \mhat_j
\end{align*}
To prove that $\mhat\o m = \id$, we use that the family $(c_i)$ is collectively epic.
Thus, we show that for all $i$, $\mhat\o m\o c_i  = c_i$.
We again consider $\bigsqcup_{j\geq i}  c_j$ only: 
\begin{align*}
  \textstyle
  \mhat\o m\o c_i
  &\textstyle=  \big(\bigsqcup_{j\geq i}  c_j\o\mhat_j\big) \o m   \o c_i\\
  &\textstyle= \bigsqcup_{j\geq i} c_j\o\mhat_j \o m_i \\
  &\textstyle= \bigsqcup_{j\geq i} c_j\o\mhat_j \o m_j \o e_{i,j}\\
  &\textstyle=\bigsqcup_{j\geq i} c_j \o e_{i,j} \\
  &\textstyle= \bigsqcup_{j\geq i} c_i \\
  &=c_i.
\end{align*}
In the other direction, we show that $m\o \mhat \sqsubseteq \id$:
\[\textstyle
  m \o \big( \bigsqcup_i c_i\o\mhat_i\big)
  = 
  \bigsqcup_i m\o c_i \o \mhat_i
  =
  \bigsqcup_i m_i \o \mhat_i \sqsubseteq \id.
\]
Our last order of business is to show that 
$m =\bigsqcup_{i\in A} m_i$ in $\Sub_{\E}(A)$.
Since $m \o c_i = m_i$, we see that $m_i \sqsubseteq m$ for all $i$.
Let $u\colon U\monoto A$ be an embedding with  $m_i \sqsubseteq u$  for all $i$.
Thus we have morphisms $u_i$ such that  $m_i = u \o u_i$.
The family $(u_i)_i$ is a cocone of the original diagram $D$, because if $i\leq j$, then
\[
  u_j \o e_{i,j}
  =
  (\widehat{u}\o u) \o u_j \o e_{i,j}
  =
  \widehat{u} \o m_j \o e_{i,j}
  =
  \widehat{u} \o m_i
  =
  \widehat{u} \o  u \o u_i = u_i.
\]
Since $(c_i)$ is a colimit, there is a unique $f\colon M \to U$ such
that $u_i = f \o c_i$ for all $i$.  We aim to show that $m = u \o f$,
so that $m \sqsubseteq u$.  For this, we again use the fact that
$(c_i)$ is a collectively epic family:
$ m \o c_i = m_i = u \o u_i = u \o f \o c_i $.
\end{proof}

\end{document}